\documentclass[letterpaper, 10 pt, conference]{ieeeconf}  

\usepackage{color,array}

\usepackage{amssymb,amsfonts}

\usepackage{graphicx}
\usepackage{verbatim}   
\usepackage{subcaption}
\usepackage{algorithm}
\usepackage{algpseudocode}

\usepackage{mathtools}
\usepackage{siunitx}
\usepackage{tabularx}
\usepackage{pgfplots}
\usepackage{pgfplotstable}

\usepackage{multirow}
\usepgfplotslibrary{fillbetween}
\usepackage{booktabs}
\pgfplotsset{compat = newest}
\usetikzlibrary{arrows,positioning,shapes,intersections,external,patterns,calc,fit,decorations,decorations.markings} 
\usepackage{cite}
\newtheorem{defn}{Definition}
\newtheorem{rem}{Remark}

\newtheorem{prop}{Proposition}

\newtheorem{assum}{Assumption}

\newcommand\tran{\mkern-2mu\raise1.25ex\hbox{$\scriptscriptstyle\top\hspace{0.5mm}$}\mkern-3.5mu}
\newcommand{\R}{\mathbb{R}}

\newcommand{\bm}[1]{{\boldsymbol{#1}}}

\DeclareMathOperator{\var}{var}

\newcommand{\z}{\bm z}

\newcommand{\x}{\bm x}

\usepackage[noabbrev]{cleveref} 
\crefname{rem}{Remark}{Remarks}
\crefname{exam}{Example}{Examples}
\crefname{assum}{Assumption}{Assumptions}
\crefname{prop}{Proposition}{Propositions}
\crefname{propy}{Property}{Properties}
\crefname{cor}{Corollary}{Corollaries}
\crefname{lem}{Lemma}{Lemmas}
\crefname{section}{Section}{Sections}
\crefname{thm}{Theorem}{Theorems}
\crefname{alg}{Algorithm}{Algorithms}
\crefname{defn}{Definition}{Definitions}
\crefname{figure}{Fig.}{Fig.}
\Crefname{figure}{Figure}{Figures}
\crefname{equation}{}{}

\IEEEoverridecommandlockouts                              

\title{\LARGE \bf
Data-Driven Boundary Control of Distributed Port-Hamiltonian Systems
}

\author{Thomas Beckers and Leonardo Colombo
\thanks{Thomas Beckers is with the Department of Computer Science, Vanderbilt University, Nashville, TN 37212, USA {\tt\small thomas.beckers@vanderbilt.edu}}%
\thanks{L. Colombo is with Centre for Automation and Robotics (CSIC-UPM), Ctra. M300 Campo Real, Km 0,200, Arganda
del Rey - 28500 Madrid, Spain.{\tt\small leonardo.colombo@csic.es}} 
\thanks{ L. Colombo acknowledge financial support from Grant PID2022-137909NB-C21 funded by MCIN/AEI/ 10.13039/501100011033. The research leading to these results was supported in part by iRoboCity2030-CM, Robótica Inteligente para Ciudades Sostenibles (TEC-2024/TEC-62), funded by the Programas de Actividades I+D en Tecnologías en la Comunidad de Madrid.}
}

\begin{document}

\maketitle
\thispagestyle{empty}
\pagestyle{empty}

\begin{abstract}
Distributed Port-Hamiltonian (dPHS) theory provides a powerful framework for modeling physical systems governed by partial differential equations and has enabled a broad class of boundary control methodologies. Their effectiveness, however, relies heavily on the availability of accurate system models, which may be difficult to obtain in the presence of nonlinear and partially unknown dynamics. To address this challenge, we combine Gaussian Process distributed Port-Hamiltonian system (GP-dPHS) learning with boundary control by interconnection. The GP-dPHS model is used to infer the unknown Hamiltonian structure from data, while its posterior uncertainty is incorporated into an energy-based robustness analysis. This yields probabilistic conditions under which the closed-loop trajectories remain bounded despite model mismatch. The method is illustrated on a simulated shallow water system.
\end{abstract}

\section{Introduction}
The modeling and control of physical systems governed by partial differential equations (PDEs) is a crucial task in a broad range of domains such as physics, engineering, and applied mathematics~\cite{derler2011modeling}. Applications range from the control of flexible structures to fluid dynamics and wave propagation phenomena. A large and powerful class of these physical systems can be described by distributed Port-Hamiltonian systems (dPHS), see~\cite{van2002hamiltonian, macchelli2004port}.

Distributed PHSs are a powerful mathematical framework for modeling and analyzing physical systems with a spatial domain. The dPHS approach describes physical systems abstracting from detailed equations to focus on energy interactions and conservation principles. The dynamics of the system is described by a Hamiltonian functional, while the boundary interactions are described by power-based port variables. This unique approach allows for systematic analysis of complex physical systems, providing a clear understanding of their behavior and enabling the design of efficient control strategies. Thus, in recent years, the boundary control of dPHS has gained significant attention in research, see \cite{van2000l2, le2005dirac,schoberl2012casimir}.

A prominent control approach is the control by interconnection via Casimir generation (energy-Casimir method) developed for stabilization~\cite{macchelli2004port,rodriguez2001stabilization}. The result is an energy-balancing passivity-based controller that shapes the closed-loop energy function to introduce a minimum in a desired configuration. However, this method is fundamentally limited by the so-called ``dissipation obstacle'' when dealing with internal energy dissipation (such as internal friction in the shallow water equations). When this happens, it is impossible to shape the energy of the system along the directions in which dissipation is present. To solve this, recent control research has developed generalized methodologies to define new passive outputs for the dPHS. Interconnecting the boundary controller to this new passive output removes the limitations, allowing the control by interconnection to overcome the dissipation obstacle and achieve stability~\cite{macchelli2015control}.

\begin{figure}[t]
\begin{center}
	\includegraphics[width=\columnwidth]{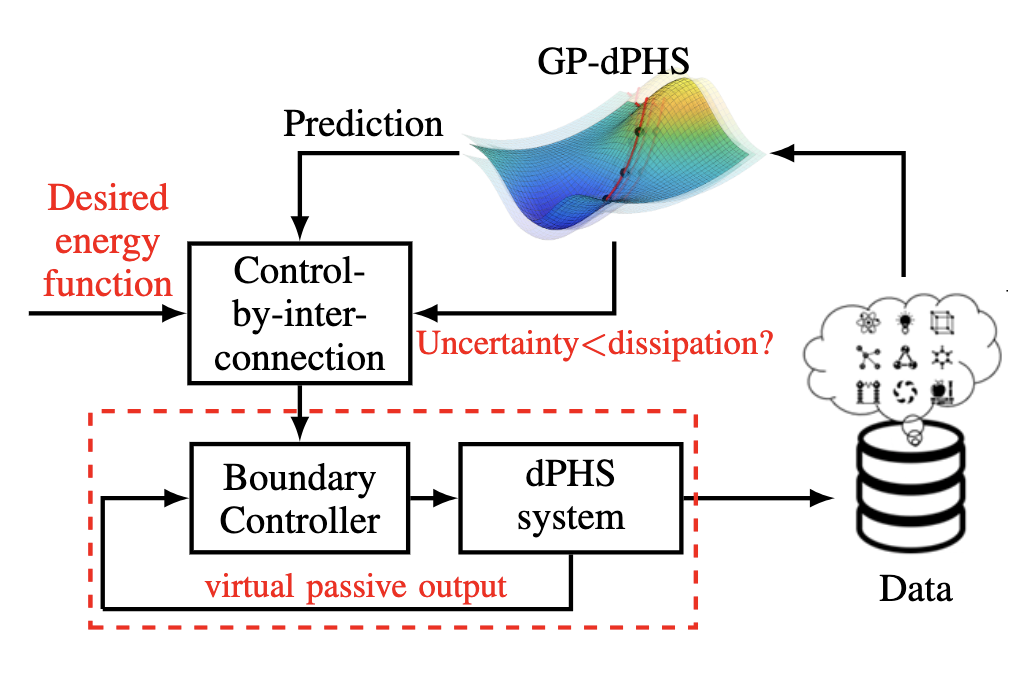}
	\vspace{-0.1cm}\caption{Overview of the control-by-interconnection approach. A GP-dPHS model is used to learn the partially unknown dynamics of the PDE systems. Then, the model enables the design of a boundary controller if the system dissipation dominates the model's uncertainty.}\vspace{-0.5cm}
	\label{fig:over}
\end{center}
\end{figure}

Despite its effectiveness, these model-based control methods require an accurate model of the system's Hamiltonian functional, which can be challenging to obtain due to highly nonlinear dynamics or the complexity of the physical system. To address these limitations, recent research explores data-driven and machine learning techniques. Recently, Gaussian process distributed Port-Hamiltonian systems (GP-dPHS) were introduced as Bayesian data-driven techniques for the identification of physical systems, see~\cite{9992733, tan2024physics, beckers2025physics}. Gaussian processes (GPs) provide a probabilistic framework for modeling complex and noisy data, enabling not only predictions but also uncertainty quantification. The goal of this paper is to combine the power of boundary control by interconnection approaches with data-driven Bayesian GP-dPHS models to enable the control of nonlinear PDE systems with partially unknown dynamics.

\textbf{Contribution:} In this paper, we propose a data-driven Bayesian control approach for physical PDE systems with partially unknown dynamics. For this purpose, we employ Gaussian Process distributed Port-Hamiltonian system models to learn the unknown dynamics of the system based on observations. We then introduce a boundary controller based on control by interconnection, which uses the mean prediction of the GP-dPHS model to render the closed-loop dynamics to shape the desired energy, overcoming any dissipation obstacles via new passive outputs. The uncertainty quantification of the GP-dPHS model allows us to consider model mismatch and derive conditions for boundedness of the closed-loop. Finally, we evaluate the closed-loop performance through the identification and control of the shallow water equations acting as a nonlinear PDE.

The remainder of the paper is structured as follows. We introduce the core concepts and the problem setting in~\cref{sec:def}, followed by the proposed GP-dPHS control methodology in~\cref{sec:ctrl}. Finally, a simulation evaluation on the shallow water equations presents the benefits of our approach in~\cref{sec:sim}.

\section{Preliminaries}\label{sec:def}
In this section, we briefly introduce Gaussian Process regression, the class of distributed Port-Hamiltonian systems and the problem setting. 
\subsection{Gaussian Process Regression}\label{sec:GPIntro}
Gaussian Processes (GPs) are a class of nonparametric models that are widely used for regression and prediction tasks \cite{rasmussen2006gaussian}. 
In contrast to parametric approaches, GPs do not assume a fixed functional form, which allows them to flexibly capture complex relationships in the data. 
A key advantage of GP models is their ability to provide not only predictions but also a principled quantification of uncertainty. 
This feature is particularly valuable in practical scenarios where only limited data are available, as GPs can still achieve good performance by exploiting structural assumptions encoded in the kernel function.

A Gaussian Process is fully specified by a mean function $m(\cdot)$ and a covariance (kernel) function $k(\cdot,\cdot)$. 
For any pair of inputs $\x\in\R^n$ and $\x'\in\R^n$, the kernel $k(\x,\x')$ defines the covariance between the corresponding function values. 
Prior to observing data, a GP defines a distribution over functions, $f \sim \mathcal{GP}(m(\x), k(\x,\x'))$.

Let $X = \{x^{(1)}, x^{(2)}, \dots, x^{(N)}\}$ denote the input data and $Y = \{y^{(1)}, y^{(2)}, \dots, y^{(N)}\}$ the corresponding outputs, forming the dataset $\mathcal{D}$. 
The observations are typically assumed to be corrupted by independent Gaussian noise with zero mean and variance $\sigma_n^2\in\R_+$. Conditioning on the dataset $\mathcal{D}$, the posterior distribution at a test input $x^*\in\R^n$ is again Gaussian, with mean and covariance given by
\begin{align}
    \mu(x^*) &\!=\! K(x^*, X)\bigl[K(X,X) + \sigma_n^2 I\bigr]^{-1} Y, \label{for:PostMean}\\
    \var(x^*) &\!=\! k(x^*,x^*) \!-\! K(x^*,X)\bigl[K(X,X)\! +\! \sigma_n^2 I\bigr]^{-1}K(X,x^*). \label{for:PostVar}
\end{align}

Here, $K(X,X)\in\R^{N\times N}$ denotes the covariance matrix of the training inputs, while $K(x^*,X)\in\R^{1\times N}$ and $K(X,x^*)\in\R^{N\times 1}$ denote the cross-covariances between the test input and the training inputs. Moreover, $k(x^*,x^*)$ represents the prior variance at $x^*$. 
The resulting predictive distribution $p(y^* \mid x^*, \mathcal{D})$ therefore provides both a point estimate and an associated uncertainty measure, which is a defining characteristic of GP regression.

\subsection{Distributed Port-Hamiltonian systems}
Composing Hamiltonian systems with input/output ports leads to a so-called Port-Hamiltonian system, which is a dynamical system with ports that specify the interactions of its components. In~\cite{macchelli2004port}, the classical finite-dimensional Port-Hamiltonian formulation of a dynamical system is generalized in order to include the distributed-parameter and multivariable case. In contrast to finite-dimensional PHS, the interconnection, damping, and input/output matrices are replaced by matrix differential operators that are assumed to be constant, i.e., no explicit dependence on the state (energy) variables is considered. As in finite dimensions, given the Stokes--Dirac structure, the model of the system easily follows once the Hamiltonian function is specified. 

The resulting class of infinite-dimensional systems in Port-Hamiltonian form is quite general, thus allowing the interpretation of many classical PDEs within this framework and the description of several physical phenomena, such as heat conduction, piezoelectricity, and elasticity. 

In this paper, we follow~\cite{macchelli2015control} and focus on the class of distributed port-Hamiltonian systems described by the following PDE
\begin{equation}
\frac{\partial x}{\partial t}(t,z) = \left(P_1 \frac{\partial}{\partial z}  + (P_0 - G_0)\right)\frac{\delta H}{\delta x}(x(t,z)).
\label{for:pch}
\end{equation}
Such a class generalizes what has been studied in \cite{le2005dirac, jacob2012linear}, and includes common PDE systems such as the shallow water equation. Here, the spatial domain is $\mathcal{Z} = [a,b]$ and $x(t,\cdot) \in  L^2(a,b;\mathbb{R}^n)$ denotes the state (energy) variable. Moreover, $P_1 = P_1^\top > 0$, $P_0 = -P_0^\top$, and $G_0 = G_0^\top \geq 0$, while $H$ is the Hamiltonian functional (e.g., the total energy) of the system, which is not necessarily quadratic in the energy variables. Finally, $\delta H/\delta x$ denotes the variational derivative of $H$.

To define a distributed port-Hamiltonian system, the PDE~\eqref{for:pch} has to be completed with a boundary port. In this respect, the boundary port variables associated to~\eqref{for:pch} are the vectors $f_{\partial}, e_{\partial} \in \mathbb{R}^n$ defined by
\begin{equation}
\label{for:boundaryport}
\begin{pmatrix}
f_{\partial} \\
e_{\partial}
\end{pmatrix}
=
\frac{1}{\sqrt{2}}
\begin{pmatrix}
P_1 & -P_1 \\
I & I
\end{pmatrix}
\renewcommand*{\arraystretch}{1.2}
\begin{pmatrix}
\frac{\delta H}{\delta x}(b) \\
\frac{\delta H}{\delta x}(a)
\end{pmatrix}.
\end{equation}

These variables are linear combinations of the boundary traces of the co-energy variables. To obtain a boundary control system, inputs and outputs have to be defined. Let $W$ and $\widetilde{W}$ be a pair of $n \times 2n$ full-rank real matrices such that $(W^\top \ \widetilde{W}^\top)$ is invertible, and
\begin{equation}
W \Sigma W^\top = 0, \quad
W \Sigma \widetilde{W}^\top = I, \quad
\widetilde{W} \Sigma \widetilde{W}^\top = 0,
\end{equation}
where
\[
\Sigma =
\begin{pmatrix}
0 & I \\
I & 0
\end{pmatrix}.
\]

The boundary input $u\colon\R_+\to\R^n$ and output $y\colon\R_+\to\R^n$ can be defined as
\begin{equation}
\label{for:bound}
u(t) = W
\begin{pmatrix}
f_{\partial}(t) \\
e_{\partial}(t)
\end{pmatrix},
\quad
y(t) = \widetilde{W}
\begin{pmatrix}
f_{\partial}(t) \\
e_{\partial}(t)
\end{pmatrix}.
\end{equation}

Then, the following energy balance equation holds
\begin{equation}
\frac{d}{dt}H(x(t,\cdot)) = - \int_a^b \left( \frac{\delta H}{\delta x} \right)^\top G_0 \frac{\delta H}{\delta x} \, dz + y^\top u \leq y^\top u,
\end{equation}
where the integral term accounts for the dissipative effects along the spatial domain, while $y^\top u$ represents the supplied power at the boundary port.

\subsection{Problem Setting}\label{sec:ps}
We consider the challenge of designing a boundary control law for a partially unknown physical PDE system whose dynamics can be modeled as a distributed Port-Hamiltonian system~\cref{for:pch}. For this, we assume that the Hamiltonian functional $H$ is sufficiently smooth and \emph{only partly or completely unknown} due to unstructured nonlinearities in the system (such as nonlinear stress-strain material properties or changes in material properties over the spatial domain). To address the problem, we introduce the following assumptions.
\begin{assum}\label{assum:1}
The system matrices $P_0, P_1,$ and $G_0$ of the PDE system~\cref{for:pch} are assumed to be known from basic physics principles. 
\end{assum}
\begin{assum}\label{assum:2}
We can observe the state of the PDE system~\cref{for:pch} at some temporal instants $t_i$ and spatial points $z_j$ to generate a set of observations $\{x(t_1,z_1),\ldots,x(t_1,z_{N_z}), \ldots, x(t_{N_t},z_1),\ldots,x(t_{N_t},z_{N_z})\}$.
\end{assum}
\begin{assum}\label{assum:3}
There exists a unique and continuous solution $x(t,z)$ for the PDE system under the boundary input-output structure~\cref{for:bound}.
\end{assum}
\Cref{assum:1} is only mildly restrictive since these matrices are often known from physical prior knowledge of the system, whereas the Hamiltonian typically contains the unstructured, hard-to-model nonlinearities. However, determining the dissipation matrix $G_0$ might be challenging, but the elements of $G_0$ can be estimated following the procedure in~\cite{tan2024physics}.

\Cref{assum:2} ensures that we can collect data from the PDE system, requiring observability of the state. If that is not the case, an observer needs to be implemented. Finally, \cref{assum:3} guarantees the existence of a solution so that the derivation of a surrogate model is meaningful.

Since the controller is synthesized from a learned model and then applied to the physical PDE system, the effect of model mismatch must be explicitly taken into account.

\textbf{Goal:} Given a dataset as in~\cref{assum:2}, our aim is to design a controller that allows us to shape the Hamiltonian functional of the closed-loop system while overcoming the internal dissipation obstacle induced by the distributed dissipative structure.

\section{Control Design based on GP-dPHS}
\label{sec:ctrl}
The general idea of the proposed approach is depicted in~\cref{fig:over}. First, we collect (noisy) state measurements across the spatial and temporal domains from the physical system, which are then used to train a Gaussian Process distributed Port-Hamiltonian system model. Due to its Bayesian nature, this model provides both a prediction of the unknown physical dynamics and an explicit uncertainty quantification. The mean prediction is used to design a boundary control law based on control by interconnection via Casimir generation. To handle the internal dissipation present in systems like the shallow water equations, we exploit a new passive output that overcomes the classical dissipation obstacle. The uncertainty prediction of the GP-dPHS model can be leveraged to ensure the robustness of the control law against model mismatch. We begin with data collection and training of the GP-dPHS model, as introduced in~\cite{tan2024physics}. Afterward, we present the boundary control by interconnection method.

\subsection{Data Collection}\label{subsubsec:Data}
First, we leverage the collected data of the PDE system with unknown dynamics~\cref{for:pch} to train our GP-dPHS model. Given~\cref{assum:2}, we can define a set of observations of the system as 
\begin{align}
    \mathcal{D} = \{t_i,z_j,x(t_i,\z_j),u(t_i)\}_{i=0,j=0}^{i=N_t-1,j=N_z-1}, \label{for:dataD}
\end{align}
where \( x(t_i, z_j) \) denotes the observed state at time \( t_i \) and spatial point \( \z_j \), corresponding to \( N_t \) time steps and \( N_z \) spatial points. These observations are paired with a boundary input sequence \( \{u(t_1), u(t_2), \dots,u(t_{N_t})\} \) used to excite the system and generate a rich dataset.
\begin{rem}
As in standard system identification, the accuracy of the learned model generally improves with the availability of informative data. However, no strict assumptions are imposed on the input sequence, as the model quality is subsequently reflected in the associated uncertainty of the GP-dPHS and, in this way, included in the design of the controller.
\end{rem}

In practice, the number of observations is usually limited across the spatial domain, but we also require state derivatives for training of the GP-dPHS model. Hence, we deploy GP regression to upsample the data and to estimate the state derivatives. Using a smooth kernel function trained on the input data $\{t_i,z_j,u(t_i)\}_{i=0,j=0}^{i=N_t-1,j=N_z-1}$ and output data $\{x(t_i,z_j)\}_{i=0,j=0}^{i=N_t-1,j=N_z-1}$, we can take the derivative of the posterior mean as an estimate of the state derivative \(\frac{\partial}{\partial t} x\), which allows the dataset to be augmented with \( N_e \) additional spatial points, where \( N_e \gg N_z \). Consequently, we construct a new dataset \(\mathcal{E}\), which comprises the states \( X = [\tilde{x}(t_0), \ldots, \tilde{x}(t_{N_t-1})] \) and the state derivatives \( \dot{X} = \left[\frac{\partial \tilde{x}(t_0)}{\partial t}, \ldots, \frac{\partial \tilde{x}(t_{N_t-1})}{\partial t} \right] \), where \(\tilde{x}(t_i) = [x(t_i, z_0)^\top, \ldots, x(t_i, z_{N_e-1})^\top]^\top\) represents the spatially stacked state at time \( t_i \). In this way, we define the augmented dataset \(\mathcal{E} = [X, \dot{X}]\). For further details, see~\cite{tan2024physics}.

\subsection{Training of the Gaussian Process distributed Port-Hamiltonian model} \label{subsubsec:trainingGPdPHS}
Next, since the exact underlying dynamics are unknown, we leverage the dataset \(\mathcal{E}\) to learn a distributed Port-Hamiltonian model of the PDE system. This approach uses a Gaussian process to approximate the unknown Hamiltonian functional $H$ in~\cref{for:pch}, which is equivalent to learning an energy representation of the PDE system. In practice, the GP-dPHS model is trained on the finite-dimensional surrogate representation induced by the spatially stacked state \(\tilde{\x}(t_i)\). This construction is made possible because GPs are invariant under linear transformations, see~\cite{9992733}, allowing us to incorporate the Hamiltonian functional derivatives into the GP framework. This leads to the following GP representation of the dPHS dynamics:
\begin{align}
    \frac{\partial \x}{\partial t} \sim \mathcal{GP}(0, k_{dphs}(\x, \x')), \label{gp:frac}
\end{align}
where the system dynamics are encapsulated within the probabilistic model. Consequently, we define a new physics-informed kernel function \( k_{dphs} \) under the dPHS formalism, which takes the form
\begin{align}
    k_{dphs}(\x, \x') = \sigma^2_f\widehat{JR} \delta_\x \exp\left(-\frac{\|\x - \x'\|^2}{2 \varphi_l^2}\right)\delta_{\x'}^\top \widehat{JR}^\top,
\end{align}
where $\widehat{JR} = \left(P_1 \frac{\partial}{\partial z}  + (P_0 - G_0)\right)$. At the implementation level, the operators involved in \(k_{dphs}\) are evaluated on the finite-dimensional surrogate representation induced by the sampled spatial discretization. This kernel is based on the squared exponential function to represent the smooth Hamiltonian functional. The model is trained as a standard GP using the dataset \(\mathcal{E}\). The hyperparameters \( \varphi_l\in\R_{>0} \) and \( \sigma_f\in\R_{>0} \) are optimized by minimizing the negative log marginal likelihood. This completes the training of the GP-dPHS model. Utilizing the invariance of GPs under linear transformations, we can write the dynamics of the GP-dPHS model as
\vspace{-5pt}
\begin{align}\label{dPHSform}
\frac{\partial \hat{x}}{\partial t}(t,z) = \left(P_1 \frac{\partial}{\partial z}  + (P_0 - G_0)\right)\frac{\delta \mu(\hat{H}|\mathcal{E})}{\delta \hat{x}}(\hat{x}(t,z)),
\end{align}
where \(\hat{H}\) denotes the learned Hamiltonian functional endowed with the GP posterior distribution, and \(\mu(\hat{H}|\mathcal{E})\) denotes its posterior mean. As a consequence, we obtain the model~\cref{dPHSform} for the unknown PDE system based on the training dataset~\(\mathcal{E}\). For further details, see~\cite{tan2024physics}.
\subsection{Boundary Control by Interconnection}
In this section, we propose a data-driven boundary control approach for the (partially) unknown PDE system~\cref{for:pch}. Inspired by~\cite{macchelli2015control}, we consider a control design using the concept of control-by-interconnection. For this purpose, the following finite-dimensional control system in PHS form is considered
\begin{equation}
\dot{x}_c(t) = J_c \frac{\partial H_c}{\partial x_c}(x_c(t)) + G_c u_c(t), \,\,
y_c(t) = G_c^\top \frac{\partial H_c}{\partial x_c}
\label{eq:controller}
\end{equation}
where $x_c \in \mathbb{R}^{n_c}$ denotes the controller state, $J_c = -J_c^\top\in\R^{n_c\times n_c}$ is the interconnection matrix, $G_c\in\R^{n_c\times n_c}$ is the input matrix, and $H_c$ is the controller Hamiltonian to be designed. Since we do not know the Hamiltonian of the PDE system~\cref{for:pch}, the controller is designed based on the model~\cref{dPHSform} and then applied to the actual system. Thus, the next steps describe the controller design, followed by the analysis of connecting the controller to the actual system~\cref{for:pch}.

The controller is interconnected with the boundary port $(u,y)$ of the model~\cref{dPHSform} through a power-preserving interconnection given by
\begin{equation}
\begin{pmatrix}
u \\
y
\end{pmatrix}
=
\begin{pmatrix}
0 & -I \\
I & 0
\end{pmatrix}
\begin{pmatrix}
u_c \\
y_c
\end{pmatrix}
+
\begin{pmatrix}
u' \\
0
\end{pmatrix},
\label{eq:interconnection}
\end{equation}
where $u'$ represents a new external input, leading to a closed-loop system whose total Hamiltonian is
\begin{align}
H_d(x,x_c) = \mu(\hat{H}(x)|\mathcal{E}) + H_c(x_c).
\end{align}
By appropriately selecting $H_c$, the energy function $H_d$ can be shaped so that it attains a minimum at a desired equilibrium.

As in the finite-dimensional framework, the design process becomes more tractable by introducing invariant quantities, commonly referred to as Casimir functions. These invariants establish a relation between the plant and controller states that is independent of the particular choice of Hamiltonians.

\begin{defn}
Consider the autonomous ($u' = 0$) closed-loop system obtained from \eqref{eq:interconnection}. A function $C \colon X \times \mathbb{R}^{n_c} \to \mathbb{R}$ is called a Casimir function if it remains constant along all trajectories of the system, i.e. $\dot{C}(x(t,\cdot), x_c(t)) = 0 $
for any admissible choice of $\mu(\hat{H}|\mathcal{E})$ and $H_c$.
\end{defn}

The invariance of $C$ implies the existence of an intrinsic relation between $x_c$ and $x$ that does not depend on the specific Hamiltonian functions. Based on this definition, we note that any Casimir function of the learned model~\cref{dPHSform} is also a Casimir function of the PDE system~\cref{for:pch}.
\begin{prop}\label{prop:1}
Let $C \colon X \times \mathbb{R}^{n_c} \to \mathbb{R}$ be a Casimir function of the autonomous closed-loop system obtained from the controller~\cref{eq:controller} and \textit{model}~\cref{dPHSform}, connected via~\eqref{eq:interconnection}. Then, $C$ is also a Casimir function for the autonomous closed-loop system obtained from the controller~\cref{eq:controller} and \textit{system}~\cref{for:pch}, connected via~\eqref{eq:interconnection}.
\end{prop}
\begin{proof}
The claim follows from the structural nature of the Casimir condition, which depends on the interconnection relations and not on the particular choice of Hamiltonian. Since the learned model~\cref{dPHSform} and the actual system~\cref{for:pch} share the same interconnection structure, any Casimir function for the former is also a Casimir function for the latter.
\end{proof}
In the following, we focus on Casimir functions of the form
\begin{equation}
C(x(t,\cdot), x_c(t)) = \Gamma^\top x_c(t) + \Psi(x(t,\cdot)),
\label{eq:casimir}
\end{equation}
where $\Gamma \in \mathbb{R}^{n_c}$ and $\Psi : X \to \mathbb{R}$ is a functional of $x$.
The knowledge of Casimir functions can then be exploited to shape the closed-loop energy. However, traditional energy-shaping passivity-based control approaches encounter the fundamental ``dissipation obstacle'', not being able to deal with equilibria that require an infinite amount of supplied energy in steady state.

However, as detailed by~\cite{macchelli2015control}, this obstacle can be overcome functionallyvia a new passive output $\bar{y}$ such that the learned model remains passive with respect to the new input/output pair, namely $\frac{d}{dt}\mu(\hat{H}\mid E)\le \bar{y}^\top u$. The closed-loop system obtained from the augmented interconnection matching 
\begin{equation}
\begin{pmatrix}
u \\
\bar{y}
\end{pmatrix}
=
\begin{pmatrix}
0 & -I \\
I & 0
\end{pmatrix}
\begin{pmatrix}
u_c \\
y_c
\end{pmatrix}
+
\begin{pmatrix}
u' \\
0
\end{pmatrix},
\label{eq:new_interconnection}
\end{equation}
which replaces $y$ by the new output $\bar{y}$, for the model \eqref{dPHSform} and the controller \eqref{eq:controller}, is characterized by a set of Casimir functions of the form~\cref{eq:casimir}, where each functional $\Psi$ satisfies
\begin{equation}
P_1 \frac{\partial}{\partial z} \frac{\delta \Psi}{\delta x}(x)
+ (P_0 - G_0)\frac{\delta \Psi}{\delta x}(x) = 0.
\label{eq:new_casimir_pde}
\end{equation}
Next, we determine the new passive output $\bar{y}$ such that the controller can act on the Hamiltonian $\hat{H}$ without being restricted by the dissipation obstacle.
\begin{prop}[adapted from~\cite{macchelli2015control}]\label{prop:2}
Consider the GP-dPHS model with dynamics given by \eqref{dPHSform}, boundary input $u$ defined in \eqref{for:bound}, the controller \eqref{eq:controller}, and the power-conserving interconnection \eqref{eq:new_interconnection}. Then, the function \eqref{eq:casimir} is a Casimir invariant, with $\Psi$ satisfying \eqref{eq:new_casimir_pde}, if
\begin{align}
\bar{y}(t) &= y(t) + S u(t)\label{eq:ybar}\\
&+ 2 \frac{G_c^\top \Gamma}{\|G_c^\top \Gamma\|^2}
\!\int_a^b\! 
\frac{\delta \Psi}{\delta x}^\top(x(t,z))\, G_0 \,
\frac{\delta \mu(\hat{H}|\mathcal{E})}{\delta x}(x(t,z)) \, dz,\notag
\end{align}
where $y$ is the original output defined in \eqref{for:bound}, and
\begin{equation}
S = \frac{G_c^\top \Gamma}{(b-a)\|G_c^\top \Gamma\|^4}
\left(
\int_a^b 
\frac{\delta \Psi}{\delta x}^\top G_0 \frac{\delta \Psi}{\delta x} \, dz
\right)
\Gamma^\top G_c + S',
\label{eq:S}
\end{equation}
with $S' = (S')^\top \geq 0$ and the controller \eqref{eq:controller} satisfies
\begin{align}
\begin{split}
\left(J_c + G_c S G_c^\top \right)\Gamma
+ G_c \tilde{W} R
\renewcommand*{\arraystretch}{1.2}
\begin{bmatrix}
\frac{\delta \Psi}{\delta x}(b) \\
\frac{\delta \Psi}{\delta x}(a)
\end{bmatrix}
&= 0, \label{eq:prop7_cond1} \\
G_c^\top \Gamma
+ W R
\renewcommand*{\arraystretch}{1.2}
\begin{bmatrix}
\frac{\delta \Psi}{\delta x}(b) \\
\frac{\delta \Psi}{\delta x}(a)
\end{bmatrix}
&= 0. 
\end{split}
\end{align}
\end{prop}
\begin{proof}
The proof follows that of Proposition 7 in~\cite{macchelli2015control}, replacing the system Hamiltonian \(H\) by the learned mean Hamiltonian \(\mu(\hat{H}\mid \mathcal{E})\).
\end{proof}
In this way, with the new virtual output $\bar{y}$~\cref{eq:ybar} and a controller~\cref{eq:controller}, satisfying~\cref{eq:prop7_cond1}, we can shape the closed-loop Hamiltonian $H_d$ via $H_c$ so that it attains a minimum at a desired equilibrium.

Next, we analyze the consequences of connecting the designed controller~\cref{eq:controller} with the actual PDE system~\cref{for:pch} instead of the model~\cref{dPHSform}. 
For this, the dynamics of the closed-loop system is written as
\begin{align}
\begin{split}
\label{eq:actual_closedloop}
\frac{\partial x}{\partial t}(t,z) &\!=\! \left(\!P_1 \frac{\partial}{\partial z} + (P_0 - G_0)\!\right)\frac{\delta \mu(\hat{H}\mid\mathcal{E})}{\delta x}(x(t,z))\! +\! \eta(x),\\
\dot{x}_c(t) &= J_c \frac{\partial H_c}{\partial x_c}(x_c(t)) + G_c \bar{y}(t), \\
u(t) &= -G_c^\top \frac{\partial H_c}{\partial x_c}(x_c(t)),
\end{split}
\end{align}
where $\mu(\hat{H}\mid\mathcal{E})$ is the posterior mean of the learned Hamiltonian and $\eta(x)$ is the unknown model error.

Before stating the main result, we impose the following assumption.
\begin{assum}\label{ass:bound}
There exists a probability $p\in(0,1)$ and a constant $\bar{\eta}\in\R_+$ such that the model error satisfies $\|\eta(x)\|_{L^2(a,b)} \le \bar{\eta}$ with probability at least $1-p$ for all admissible states $x$.
\end{assum}

\Cref{ass:bound} is a standard robustness-type assumption in learning-based control and expresses that, with high probability, the residual model mismatch remains uniformly bounded on the domain of interest.
\begin{prop}\label{prop:3}
Consider the partially unknown PDE system~\cref{for:pch}, a GP-dPHS model~\cref{dPHSform} trained on observed data \(\mathcal{D}\) satisfying~\cref{ass:bound}, and a controller \eqref{eq:controller} following~\cref{prop:2}, which is connected with~\cref{for:pch} via \cref{eq:new_interconnection}. Let \(H_d\) be coercive, and assume that there exists \(\lambda>0\) such that $v^\top G_0 v \ge \lambda \|v\|^2$ for all admissible \(v\). Then, for any \(\varepsilon\in(0,2\lambda)\), the closed-loop energy satisfies
\[
\frac{dH_d}{dt} \le -\left(\lambda-\frac{\varepsilon}{2}\right)\left\|\frac{\delta \mu(\hat{H}\mid\mathcal{E})}{\delta x}\right\|_{L^2(a,b)}^2 + \frac{1}{2\varepsilon}\bar{\eta}^2
\]
with probability at least \(1-p\). Consequently, the state \(x(t,z)\) is ultimately bounded with probability at least \(1-p\), and the trajectories evolve inside a sublevel set of \(H_d\).
\end{prop}

\begin{proof}
Define $\displaystyle{e(x):=\frac{\delta \mu(\hat{H}\mid\mathcal{E})}{\delta x}(x)}$. Differentiating \(H_d\) along the closed-loop dynamics yields
\begin{align}
    \frac{dH_d}{dt} &= - \int_a^b e(x)^\top G_0 e(x) \, dz \notag\\
    &\quad + \int_a^b e(x)^\top \left(P_1 \frac{\partial}{\partial z} + P_0 \right)e(x) \, dz \notag\\
    &\quad + \int_a^b e(x)^\top \eta(x) \, dz  + \left(\frac{\partial H_c}{\partial x_c}\right)^\top J_c \frac{\partial H_c}{\partial x_c}. \label{for:proof}
\end{align}
Due to the skew-symmetry of \(J_c\) and \(P_0\), together with the power-preserving structure of the interconnection, the non-dissipative terms do not contribute to the decrease of \(H_d\). Here, we use the standard dPHS power balance under the imposed boundary interconnection, so that the conservative interconnection terms cancel in the energy identity. Hence,
\[
\frac{dH_d}{dt} \le - \int_a^b e(x)^\top G_0 e(x)\,dz + \int_a^b e(x)^\top \eta(x)\,dz.
\]
Using the coercivity of \(G_0\), we obtain
\[
\int_a^b e(x)^\top G_0 e(x)\,dz \ge \lambda \|e(x)\|_{L^2(a,b)}^2.
\]
Next, by Cauchy--Schwarz and Young's inequality, for any \(\varepsilon>0\),
\begin{align*}
\int_a^b e(x)^\top \eta(x)\,dz &\le \|e(x)\|_{L^2(a,b)}\|\eta(x)\|_{L^2(a,b)}\\
&\le \frac{\varepsilon}{2}\|e(x)\|_{L^2(a,b)}^2 + \frac{1}{2\varepsilon}\|\eta(x)\|_{L^2(a,b)}^2.
\end{align*}
Therefore,
\[
\frac{dH_d}{dt}
\le -\left(\lambda-\frac{\varepsilon}{2}\right)\|e(x)\|_{L^2(a,b)}^2
+ \frac{1}{2\varepsilon}\|\eta(x)\|_{L^2(a,b)}^2.
\]
By \cref{ass:bound}, with probability at least \(1-p\), $\|\eta(x)\|_{L^2(a,b)}^2 \le \bar{\eta}^2$,
and hence
\[
\frac{dH_d}{dt}
\le -\left(\lambda-\frac{\varepsilon}{2}\right)\|e(x)\|_{L^2(a,b)}^2
+ \frac{1}{2\varepsilon}\bar{\eta}^2.
\]
Choosing \(\varepsilon\in(0,2\lambda)\) makes the first coefficient strictly positive. Since \(H_d\) is radially unbounded, its sublevel sets are bounded, and the above estimate implies ultimate boundedness of the trajectories with probability at least \(1-p\).
\end{proof}

The proposition establishes a probabilistic ultimate boundedness result for the closed-loop distributed port-Hamiltonian system under model uncertainty. In particular, the distributed dissipation represented by \(G_0\) dominates the residual model mismatch up to a bounded perturbation level, and the resulting energy estimate guarantees that the trajectories remain confined to a bounded sublevel set of \(H_d\) with probability at least \(1-p\). In the ideal case \(\eta(x)=0\), the perturbation term vanishes and the closed-loop energy becomes non-increasing.
\section{Evaluation}
\label{sec:sim}
For the evaluation, we consider the boundary control problem of the nonlinear shallow water equation, adapted from~\cite{macchelli2015control}. Specifically, we study a rectangular open channel of length $L$ and unitary width, equipped with upstream and downstream gates. The fluid is assumed to have unit density and to be subject to linear friction. 

With the one-dimensional spatial coordinate $z \in [0,L]$, the system can be written in dPHS form as
\begin{equation}
\label{for:shwex}
\frac{\partial}{\partial t}
\begin{pmatrix}
q \\
p
\end{pmatrix}
=
\left(
\begin{pmatrix}
0 & -1 \\
-1 & 0
\end{pmatrix}
\frac{\partial}{\partial z}
-
\begin{pmatrix}
0 & 0 \\
0 & D
\end{pmatrix}
\right)
\frac{\delta H}{\delta x}(q,p),
\end{equation}
where $x = (q,p)$ and $D \ge 0$ models dissipation effects. The variable $q(t,z) > 0$ denotes the infinitesimal volume, which corresponds to the water height for a rectangular channel, while $p(t,z)$ represents the kinetic momentum density. The Hamiltonian is given by
\begin{equation}
H(q,p) = \frac{1}{2} \int_0^L \left( q p^2+\Delta(p) + g q^2 \right) dz,\label{for:ukHam}
\end{equation}
where $g$ is the gravitational constant and $\Delta(p)$ is a nonlinear term typically used to model turbulence, see~\cite{pope2001turbulent}. For illustration, we choose $\Delta(p)=\exp(-(p-5)^2)$, but assume throughout the remainder of this section that this term is unknown.

To account for this uncertainty, we can learn the Hamiltonian using a GP-dPHS approach as described in~\cref{sec:ctrl}. For simplicity, we directly train a Gaussian process with a squared exponential kernel on~\cref{for:ukHam}, using the prior mean $m(x)=q p^2 + g q^2$ and 100 training points $(q,p)$ sampled from the domain $[0,\,10]\times [-10,\,10]$. The resulting posterior mean $\mu(\hat{H}\mid\mathcal{E})$ is then used as a surrogate for~\cref{for:ukHam}. Due to the structure of the shallow water system~\cref{for:shwex} and the boundedness of the squared exponential kernel (and hence of the GP posterior mean), the dissipative term can dominate the model uncertainty.

The corresponding co-energy variables are given by $\frac{\delta \mu(\hat{H}\mid\mathcal{E})}{\delta q} =: P(q,p)$ and $\frac{\delta \mu(\hat{H}\mid\mathcal{E})}{\delta p} =: Q(q,p)$, which represent pressure and flow, respectively. The boundary input/output pair is defined as
\begin{align}
u(t) &= 
\begin{pmatrix}
Q(0,t) \\
P(L,t)
\end{pmatrix},
\quad
y(t) =
\begin{pmatrix}
P(0,t) \\
-Q(L,t)
\end{pmatrix}.
\end{align}
The equilibrium $(q^\star, p^\star)$ satisfies
\begin{align}\label{for:eqpoint}
\frac{\partial Q^\star}{\partial z} &= 0,\, \frac{\partial P^\star}{\partial z} + D Q^\star = 0,
\end{align}
with $Q^\star(z) = \bar{Q},\,P^\star(z) = \bar{Q} D (L - z) + \bar{P}$.

For the controller design, we consider a system-to-connect~\cref{eq:controller} of dimension $n_c = 2$. This choice allows us to construct the following Casimir functions in closed-loop
\begin{align}
C_1(x_{c1}, q, p) &= x_{c1} - \int_0^L [D (L - z) q + p] \, dz \; =: \Psi_1(q,p), \notag\\
C_2(x_{c2}, q, p) &= x_{c2} - \int_0^L q \, dz \; =: \Psi_2(q,p).\label{for:casinex}
\end{align}
Importantly, the derivation of these Casimir functions does not require explicit knowledge of the Hamiltonian~\cref{for:ukHam}, see~\cref{prop:1}.

Following~\cref{prop:3}, the modified passive output $\bar{y}$ used to overcome the \emph{dissipation obstacle} is given by~\cref{eq:ybar}
\begin{align}
\bar{y}(t)
&=
y(t)
-
\begin{pmatrix}
2D \displaystyle\int_0^L \frac{\delta \mu(\hat{H}\mid\mathcal{E})}{\delta p}(q(t,z),p(t,z))\,dz \notag\\
0
\end{pmatrix}\\
&+
\begin{pmatrix}
DL & 0 \\
0 & 0
\end{pmatrix}
u(t)
\tag{39}
\end{align}
where $G_c=I$, $\Gamma=-I$, $S=\begin{bmatrix}
    DL & 0\\0 & 0
\end{bmatrix}$ in~\cref{prop:2}, and $J_c=\begin{bmatrix}
    0 & 1\\-1 & 0
\end{bmatrix}$ in~\cref{eq:ybar}; see~\cite{macchelli2015control} for details.

Finally, the closed-loop Hamiltonian can be shaped by selecting $H_c$ in~\cref{eq:controller} as
\begin{align}
H_c(x_{c1}, x_{c2}) &=
\frac{1}{2}\,\Xi_1 (x_{c1} - x_{c1}^\star)^2
+
\frac{1}{2}\,\Xi_2 (x_{c2} - x_{c2}^\star)^2\notag\\
&- \bar{Q}\,x_{c1} - \bar{P}\,x_{c2}
\end{align}
with positive constants $\Xi_1, \Xi_2 \ge 0$. The reference values $x_{c1}^\star, x_{c2}^\star$ correspond to the values of the Casimir functions~\cref{for:casinex} at the equilibrium~\cref{for:eqpoint}.

\Cref{fig:gpcphs} illustrates the evolution of the water level for the closed-loop system. Starting from a constant initial condition $q(0,\cdot)=1$ (gray dashed line), the controller regulates the inflow at the left boundary such that the water level closely converges to the desired profile (black dashed line).

\begin{figure}[t]
\begin{center}
	\includegraphics[width=\columnwidth]{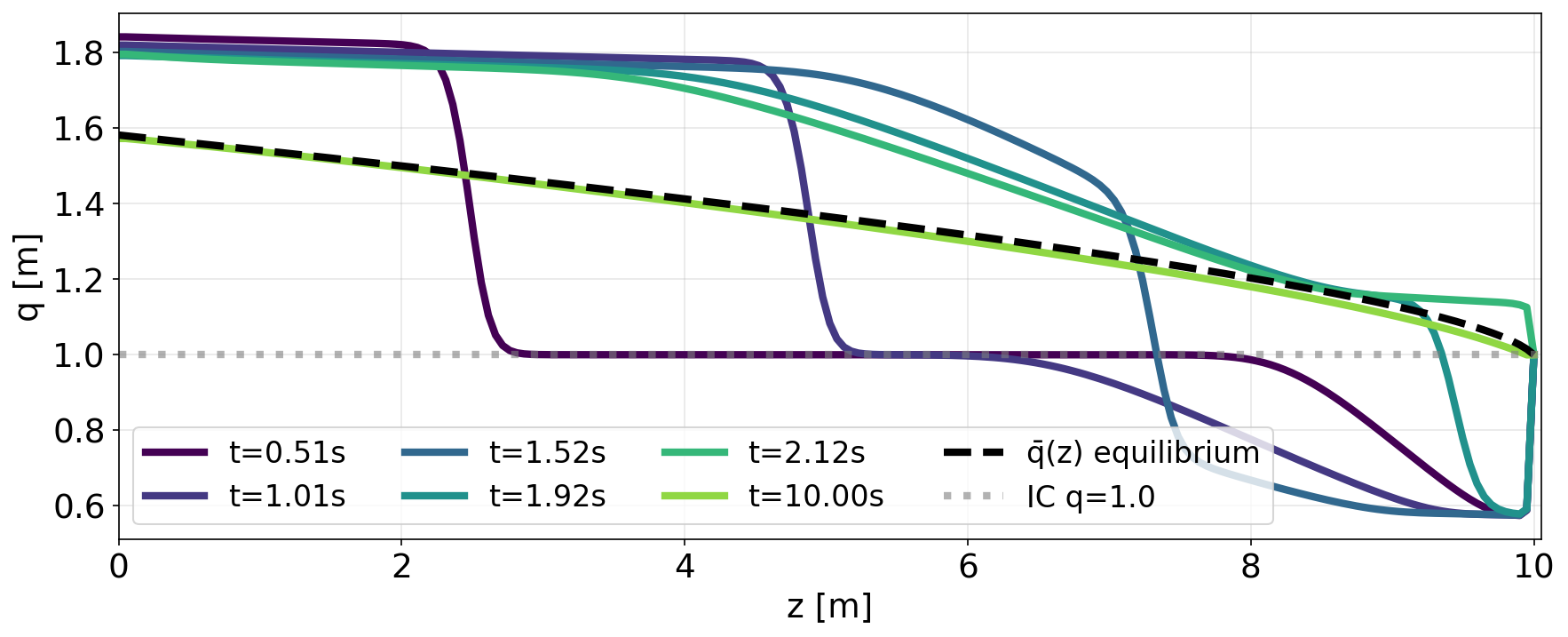}
	\vspace{-0.2cm}\caption{Controlled water level. At $t=0$, the water level is constant over the spatial domain (gray dashed). Then, the controller lets water flow into the system from the left boundary, reaching the desired equilibrium (black dashed) after 10 seconds.}\vspace{-0.7cm}
	\label{fig:gpcphs}
\end{center}
\end{figure}

The effect of using the learned model $\mu(\hat{H}\mid\mathcal{E})$ instead of the true, unknown Hamiltonian is shown in~\cref{fig:ham}. The top plot indicates that the closed-loop Hamiltonian may temporarily increase due to model mismatch, since the controller is designed using the learned model but applied to the true system. Nevertheless, \cref{prop:3} guarantees that the system trajectories remain bounded and that the Hamiltonian converges to a set.

As comparision, if the model were perfect, i.e., $\mu(\hat{H}\mid\mathcal{E})=H$, the closed-loop Hamiltonian $H_d$ would be strictly non-increasing as depicted in the bottom plot of~\cref{fig:ham}.
\begin{figure}[h]
    \centering

    \begin{subfigure}{\columnwidth}
        \centering
        \includegraphics[width=\columnwidth]{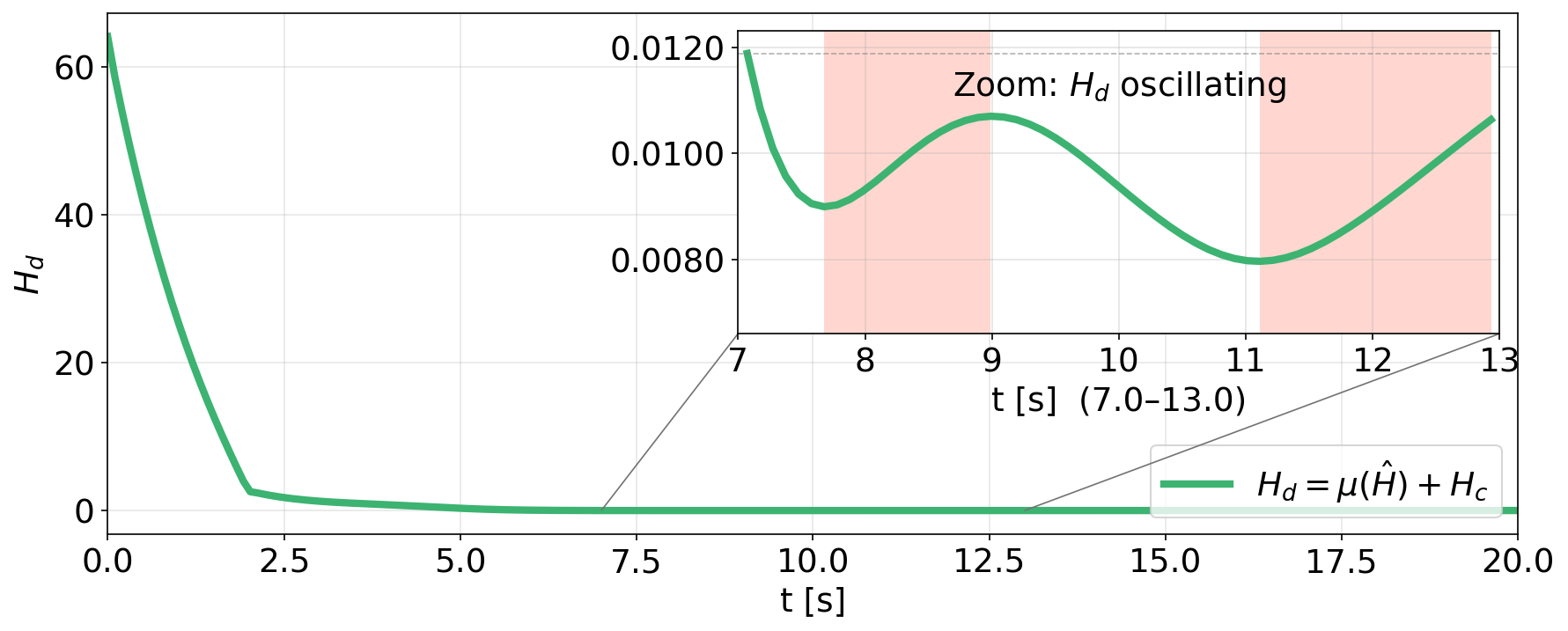}
        \caption{Controller designed based on the learned Hamiltonian $\mu(\hat{H})$.}
    \end{subfigure}

    \vspace{0.5cm}

    \begin{subfigure}{\columnwidth}
        \centering
        \includegraphics[width=\columnwidth]{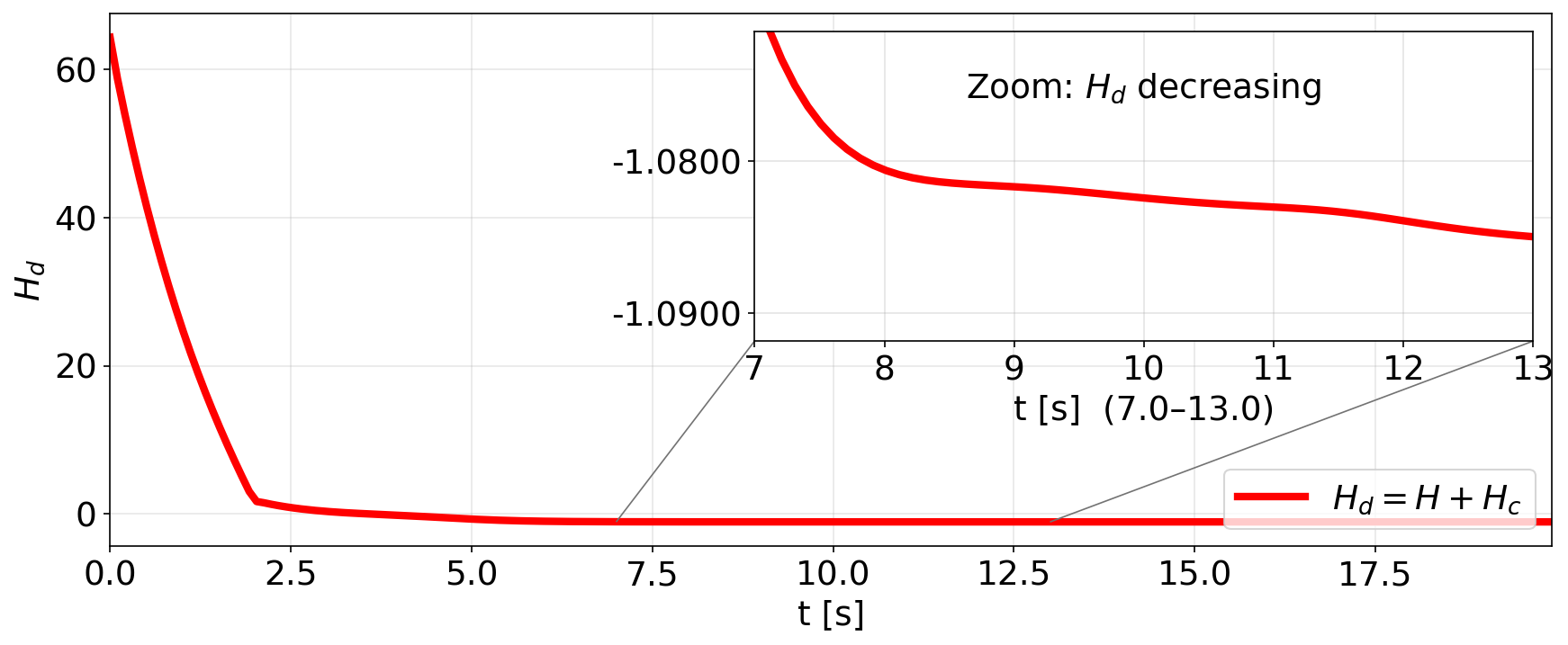}
        \caption{Controller designed based on the actual (unknown) Hamiltonian.}
    \end{subfigure}

    \caption{Comparison of the evolution of the Hamiltonian with a controller designed based on a model (a) and the actual system dynamics (b). As stated in~\cref{prop:3}, the model error leads to a Hamiltonian that converges to a set around the steady state. \label{fig:ham}}
\end{figure}
\section*{Conclusion}
In this paper, we developed a Bayesian data-driven boundary control framework for 
distributed port-Hamiltonian systems governed by PDEs with partially unknown dynamics. A GP-dPHS model enables the learning of the underlying Hamiltonian structure from localized measurements, while preserving a physically meaningful representation of the system. Based on the learned model, we designed a boundary controller by interconnection and used structurally induced Casimir relations to construct virtual passive outputs that overcome the classical dissipation obstacle. In addition, the uncertainty information provided by the GP model was incorporated into an energy-based robustness analysis, yielding probabilistic ultimate boundedness of the closed-loop dynamics under bounded model mismatch. Future work will address extensions to higher-dimensional spatial domains and more general classes of distributed port-Hamiltonian systems.

\bibliographystyle{IEEEtran}
\bibliography{root}

\end{document}